\documentclass[a4paper,10pt,dvipdfmx]{article}

\usepackage{graphicx,url}
\usepackage{epstopdf}
\usepackage{latexsym}
\usepackage{amsmath,amssymb}
\usepackage[ruled,vlined,linesnumbered]{algorithm2e}

\long\def\invis#1{}
\newenvironment{proof}{\medskip
  \noindent{\scshape Proof:}}{\quad $\Box$\medskip}

\newtheorem{prop}{Proposition}

\renewcommand{\mid}{:~}
\newcommand{\secref}[1]{Section~\ref{sec:#1}}
\newcommand{\tabref}[1]{Table~\ref{tab:#1}}
\newcommand{\figref}[1]{Fig.~\ref{fig:#1}}

\newcommand{\lineref}[1]{line~\ref{line:#1}}
\newcommand{\propref}[1]{Proposition~\ref{prop:#1}}

\newcommand{\kCol}{$k$-\textsc{Col}}
\newcommand{\VC}{\textsc{VCol}}
\newcommand{\Dsatur}{\textsc{Dsatur}}

\newcommand{\Rlf}{\textsc{Rlf}}
\newcommand{\Danger}{\textsc{Danger}}
\newcommand{\Tabucol}{\textsc{Tabucol}}
\newcommand{\Partcol}{\textsc{Partialcol}}
\newcommand{\Head}{\textsc{Head}}
\newcommand{\Foo}{\textsc{Foo}}
\newcommand{\Dyn}{\textsc{Dyn}}
\newcommand{\MS}{\mathcal{S}}
\newcommand{\Scomp}[1]{\mathcal{S}_{#1}}
\newcommand{\rhocomp}{\rho_k}
\newcommand{\Spart}[1]{\hat{\mathcal{S}}_{#1}}
\newcommand{\rhopart}{\hat{\rho}_k}

\newcommand{\MR}{\mathcal{R}}
\newcommand{\Gr}{\textsc{Grd}}
\renewcommand{\Re}{\textsc{Rec}}
\newcommand{\Rnd}{\textsc{Rnd}}
\newcommand{\Dsa}{\textsc{Dsa}}

\newcommand{\algref}[1]{Algorithm~\ref{alg:#1}}

\title{Recycling Solutions for Vertex Coloring Heuristics\thanks{The final version of the paper will appear in The Journal of the Operations Research Society of Japan.}}
\author{Yasutaka Uchida \and Kaito Yajima\and Kazuya Haraguchi\thanks{E-mail: {\tt dr.kazuya.haraguchi@gmail.com}}}
\date{}

\begin{document}
\maketitle
\begin{abstract}
  The vertex coloring problem is a well-known NP-hard problem and has many applications in
  operations research and in scheduling.
  A conventional approach to the problem solves the $k$-colorability problem iteratively,
  decreasing $k$ one by one. 
  Whether a heuristic algorithm finds a legal $k$-coloring quickly or not
  is largely affected by an initial solution.
  We highlight 
  a simple initial solution generator,
  which we call 
  the recycle method,
  which makes use of the legal $(k+1)$-coloring that has been found.  
  An initial solution generated by the method is expected
  to guide a heuristic algorithm to find a legal $k$-coloring more quickly than conventional methods,
  as demonstrated by experimental studies.
  The results suggest that the recycle method should be used as the standard
  initial solution generator for both local search algorithms
  and modern hybrid methods. 
\end{abstract}
  \invis{
    The vertex coloring problem is a well-known NP-hard problem and has many applications in
    operations research and in scheduling.
    A conventional approach to the problem solves the $k$-colorability problem iteratively,
    decreasing $k$ one by one. 
    Whether a heuristic algorithm finds a legal $k$-coloring quickly or not
    is largely affected by an initial solution.
    We highlight 
    a simple initial solution generator,
    which we call 
    the recycle method,
    which makes use of the legal $(k+1)$-coloring that has been found.  
    An initial solution generated by the method is expected
    to guide a heuristic algorithm to find a legal $k$-coloring quickly,
    as demonstrated by experimental studies.
    For example, \Tabucol, a typical local search algorithm,
    achieves $k=240$ for the {\tt R1000.5} DIMACS graph
    when the recycle method is employed as the initial solution generator.
    It is hard for most of the state-of-the-art hybrid methods to achieve this number.
  }

  \begin{center}
  \end{center}

\section{Introduction}
\label{sec:intro}
For a positive integer $n$, let $[n]=\{1,\dots,n\}$.
Assume that a graph $G=(V,E)$ is given.
We abbreviate an edge $\{u,v\}\in E$ into $uv$ for simplicity.
If $uv\in E$, we say that $u$ and $v$ are {\em adjacent},
or equivalently, that $u$ is a {\em neighbor of $v$}. 
A {\em coloring\/} is an assignment of colors to the vertices in $V$,
and it is a {\em $k$-coloring\/} if it uses at most $k$ colors. 
Representing a color by an integer,
we denote a $k$-coloring by a mapping $c:V\rightarrow[k]$.
If $c(u)=c(v)$ holds for $uv\in E$,
then we say that $uv$ is a {\em conflicting edge},
and that the extreme points $u$ and $v$ are {\em conflicting vertices\/}. 
A coloring $c$ is {\em legal} if there is no conflicting edge, and otherwise, it is illegal. 
For $i\in[k]$, let $V_{c,i}$ be the set of vertices that are given a color $i$,
that is, $V_{c,i}=\{v\in V\mid c(v)=i\}$.
We call $V_{c,i}$ the {\em color class $i$\/}.
By $c$, the vertex set $V$ is partitioned into
$V=V_{c,1}\cup\dots\cup V_{c,k}$. 
Every color class is an {\em independent set\/}
(i.e., no two vertices in the set are adjacent) iff $c$ is legal.

The {\em $k$-colorability\/} (\kCol) {\em problem\/} asks whether there exists a legal $k$-coloring for $G$. 
If the answer is yes, then we say that $G$ is {\em $k$-colorable\/}. 
The {\em vertex coloring} (\VC) {\em problem\/}
asks for a legal $k$-coloring for the smallest $k$.
The smallest $k$ that admits a legal $k$-coloring
is called the {\em chromatic number} of $G$, which we denote by $\chi(G)$.

Due to their theoretical interest as well as possible applications,
the \kCol\ problem and the \VC\ problem
have attracted many researchers from various areas
such as discrete mathematics~\cite{AW.1976,AWK.1977},
optimization~\cite{KV.2012} and scheduling~\cite{L.1979,L.2016}. 
The \kCol\ problem is NP-complete~\cite{GJ.1979}
and thus the \VC\ problem is NP-hard. 
For these problems, many exact algorithms
and heuristics have been proposed so far.
There are several booklets and surveys
on algorithmic research 
on the \VC\ problem~\cite{GHHP.2013,L.2016,MT.2010}.

When it comes to heuristics, the \VC\ problem is
often tackled by $k$-fixed strategies. 
In a $k$-fixed strategy, 
starting from an appropriate integer $k$, we search for
a legal $k$-coloring, decreasing $k$ one by one,
until a termination condition is satisfied. 
%
To make $k$ small quickly,
it is demanded to solve the \kCol\ problem as fast as possible.
Warm start should be effective for this purpose in practice~\cite{KN.2013},
that is, a ``good'' initial solution should be fed 
to the algorithm.
The quality of an initial solution is evaluated by
a penalty function that estimates how far a $k$-coloring
is from the legality (e.g., the number of conflicting edges).

In the present paper, we highlight 
a simple but effective method,
which we call the {\em recycle method\/}, 
for generating an initial solution for the \kCol\ problem.
The recycle method generates an initial solution
by modifying the legal $(k+1)$-coloring that has been found. 
More precisely, it chooses some of the $k+1$ colors
and then recolors all vertices that have the chosen colors
so that one of the $k+1$ colors disappears from the 
graph. 
Consequently, we have a $k$-coloring, 
although it must be illegal in general. 
We use it as an initial solution of the algorithm for the \kCol\ problem.

The contribution of the paper is not to propose the recycle method
but to show its empirical effectiveness.
There are various initial solution generators  
for the \kCol\ problem,
and the recycle method was already used in previous studies (e.g., \cite{CDS.2018,HK.2019}). 
However, there is no study that makes an intensive comparison of them. 


We present the background in \secref{bg},
and then explain the motivation and the detail of the recycle method
in \secref{recycle}.
Interestingly, we can derive an upper bound on the penalty value
of a $k$-coloring that is generated by the recycle method.
In \secref{comp}, we demonstrate the effectiveness of the recycle method by computational studies.
We show how the performance of three heuristic algorithms is improved when the recycle method is employed as the initial solution generator, instead of conventional
greedy/random methods, \Dsatur~\cite{B.1979} and \Rlf~\cite{L.1979}. 
The three heuristic algorithms that we use in our experiments are two well-known local search algorithms,
\Tabucol~\cite{HW.1987} and \Partcol~\cite{BZ.2008}, and
a modern hybrid method, \Head~\cite{MG.2018}. 
Specifically, with an initial solution generated by the recycle method, the heuristic algorithms
tend to find 
a legal $k$-coloring more quickly especially when $k$ is not sufficiently small. On the other hand, when $k$ is small to some extent, the difference between initial solution generators becomes comparatively small. Hence, when we do not have much computation time
and do not know a good upper bound on the chromatic number, it is a nice choice to use the recycle method as the initial solution generator. 
The benchmark instances are taken from
DIMACS graphs\footnote{\url{https://mat.gsia.cmu.edu/COLOR/instances.html}}
and Carter et al.'s timetabling instances~\cite{CLL.1996}. 
Finally we conclude the paper in \secref{conc}.

\section{Background}
\label{sec:bg}
As the \VC\ problem is NP-hard, various heuristic algorithms have been proposed so far; e.g.,
tabu search~\cite{BZ.2008,HW.1987,GPR.1996},
simulated annealing~\cite{M.1996},
variable neighborhood search~\cite{AHZ.2003},
variable space search~\cite{HPZ.2007}.
Hybrid methods include population-based 
methods~\cite{DT.2008,GH.1999,LH.2010,MMT.2008,MG.2018},
independent set extraction~\cite{HW.2012,WH.2012}, and 
quantum-inspired methods~\cite{TC.2011-2,TC.2011}.
More recently, data reduction techniques
for massive (and possibly sparse) graphs
have been studied intensively~\cite{HK.2019,LCLS.2017,VBB.2015}.
They aim at extracting a ``hard'' part of the graph by peeling out ``easy'' parts.
The hard part still needs to be solved somehow,
and thus it is a significant research issue to develop
effective heuristic algorithms
that find approximate solutions from instances with hundreds/thousands of vertices.

The search strategy of a heuristic algorithm
is determined by the definitions of the search space and
the penalty function to be minimized. 
There are two families of search strategies in the literature,
$k$-fixed ones and $k$-nonfixed ones. 
In this study, we focus on the former, which is more popular. 
Refer to \cite{GH.2006} for a detailed explanation of the two families of strategies. 

In general, a $k$-fixed strategy
searches for an approximate solution
by solving the \kCol\ problem iteratively, 
as described in \algref{k-fixed}. 

\begin{algorithm}[t!]
  \caption{An iterative scheme of $k$-fixed strategic heuristic algorithms for the \VC\ problem}
  \label{alg:k-fixed}
  \DontPrintSemicolon
  \SetKwInOut{Input}{Input}\SetKwInOut{Output}{Output}
  \Input{a graph $G=(V,E)$}
  \Output{a legal $k$-coloring for some integer $k$}
  
  Construct a legal $k$-coloring $c$ for an arbitrary $k$ somehow;
  
  \While{{\rm a termination condition is not satisfied}}{

    $k\gets k-1$;
    
    Search for a legal $k$-coloring
    (i.e., solve the \kCol\ problem)
    \label{line:kcol}
  }
  
  Output the legal $k$-coloring for the fewest $k$
  among those found in the above
\end{algorithm}

Concerning the \kCol\ problem in \lineref{kcol},
let us denote by $\MS$ the search space and  
by $\rho:\MS\rightarrow\mathbb{R}_+\cup\{0\}$ the penalty function such that $\rho(c)=0$ holds iff $c$ is legal. 
Then the \kCol\ problem is reduced to the problem
of minimizing the penalty function within $\MS$.
To reach to a legal $k$-coloring as soon as possible, if one exists,
warm start should be effective~\cite{KN.2013};
that is, to start the search from 
an initial solution
with a small penalty value. 

%
To construct the ``first'' initial solution in line 1,
we may apply various constructive algorithms in the literature
such as a random method, a greedy method (called SEQ in \cite{M.2004}),
\Dsatur~\cite{B.1979}, \Rlf~\cite{L.1979},
\Danger~\cite{GPR.1996}, and GRASP~\cite{LM.2001}.
On the other hand, less effort has been made
for generation of an initial solution in \lineref{kcol}.
A random method or a greedy method has mostly been applied for this task. 
The above-mentioned constructive algorithms can be applied,
and what we call the recycle method has already been used in \cite{CDS.2018,HK.2019}. Still, no intensive comparison has been made so far. 
%
In fact, Lewis et al.~\cite{LTMG.2012} stated that 
``the method of initial solution generation is not critical''
in \Tabucol's performance,
where \Tabucol~\cite{HW.1987} is a well-known tabu search algorithm
in the literature.
In the rest of the paper, the term ``initial solution'' refers to
one that is fed 
to a heuristic algorithm
for the \kCol\ problem in \lineref{kcol}
unless no confusion arises.

Before closing this section,
let us mention two concrete $k$-fixed strategies
that are used in the literature. 

\begin{description}
\item[$k$-fixed-penalty:]
In this strategy, the solution space is
the set of all possible $k$-colorings.
We denote the set of all $k$-colorings by $\Scomp{k}$.
A $k$-coloring $c\in\Scomp{k}$ is evaluated by
the number of conflicting edges.
Denoted by $\rhocomp:\Scomp{k}\rightarrow\mathbb{R}_+\cup\{0\}$,
it is defined to be $\rhocomp(c)\triangleq |\{uv\in E\mid c(u)=c(v)\}|$. 

\Tabucol~\cite{HW.1987} is a well-known local search algorithm
that employs the search strategy.
It explores the search space $\Scomp{k}$ by tabu search,
using $\rhocomp$ as the penalty function. 
Although it was born more than 30 years ago,
the algorithm and its extension are still
used as subroutines
in modern metaheuristics~\cite{GH.1999,GH.2006,HW.2012,LH.2010,MG.2018,WH.2012}.

\item[$k$-fixed-partial:]
The search space
is the set of what we call partial $k$-colorings.
Let us denote by $\phi$ a dummy color.
We define a {\em partial $k$-coloring} $c$
to be a function $c:V\rightarrow[k]\cup\{\phi\}$
such that, for every edge $uv\in E$,
$c(u)\ne c(v)$ holds whenever $c(u),c(v)\in[k]$.
In other words, a partial $k$-coloring admits uncolored vertices,
which are represented by $\phi$, 
but does not admit conflicting edges.
We denote the set of all partial $k$-colorings by $\Spart{k}$.
A partial $k$-coloring is evaluated by
how many vertices are assigned the dummy color. 
Denoted by $\rhopart:\Spart{k}\rightarrow\mathbb{R}_+\cup\{0\}$,
the penalty function is defined to be $\rhopart(c)\triangleq|\{v\in V\mid c(v)=\phi\}|$.

The search strategy was first introduced by Morgenstern~\cite{M.1996}. 
Bl\"ochliger and Zufferey~\cite{BZ.2008} proposed a tabu search algorithm
named \Partcol\ based on the search strategy.
\end{description}

\section{Recycle Method}
\label{sec:recycle}

In this section, we present the recycle method,
the spotlighted 
initial solution generator. 
Our aim is to develop a good initial solution generator
for {\em any} heuristic algorithms under $k$-fixed strategies.

By the recycle method,
we mean any method that constructs a $k$-coloring
from a given legal $(k+1)$-coloring, say $c$, 
as an initial solution for the \kCol\ problem. 
In $k$-fixed strategies,
a legal $(k+1)$-coloring is always available
since the $(k+1)$-\textsc{Col} problem has been already solved. 
Some of previous methods construct initial solutions from scratch (e.g., the random/greedy method),
whereas the recycle method constructs initial solutions
by making use of the legal $(k+1)$-coloring.

We describe how to generate a $k$-coloring from
a legal $(k+1)$-coloring, say $c$. 
Recall the two $k$-fixed search strategies
that we mentioned in \secref{bg}. 

\begin{description}
\item[$k$-fixed-penalty:]
  We determine a nonempty subset $K\subseteq[k+1]$ and
  an element $\varepsilon\in K$. 
  For every $i\in K$ and $v\in V_{c,i}$,  
  we change the color of $v$
  to a color in $[k+1]\setminus\{\varepsilon\}$.
  Because the color $\varepsilon$ disappears
  and there remain at most $k$ colors, 
  we have a $k$-coloring by
  degenerating the region to $[k]$. 

\item[$k$-fixed-partial:]
  We determine a nonempty subset $K\subseteq[k+1]$.
  For every $i\in K$ and $v\in V_{c,i}$,
  we assign $\phi$ to $v$.
  Because there remain at most $k$ colors
  and a dummy color $\phi$
  and no conflicting edge exists,
  we have a partial $k$-coloring
  by degenerating the region to $[k]\cup\{\phi\}$. 
\end{description}

We have freedom of designing detailed configurations
of the recycle method. 
In both strategies, the subset $K$ may be chosen randomly,
or a set of colors whose color classes are
the 
smallest.
In the $k$-fixed-penalty,
the color $\varepsilon$ to be removed is chosen at random
or can be a color whose class is
the 
smallest among $K$. 
We can change the color of a vertex in $V_{c,i}$ $(i\in K)$
into one in $[k+1]\setminus\{\varepsilon\}$ arbitrarily.

Let us describe the
motivations 
of the recycle method. 
For the \kCol\ problem, it is known that 
easy-hard-easy phase transition exists
with respect to $k$~\cite{HH.2002,M.2018}. 
The peak of difficulty is said to lie around
the chromatic number $k=\chi(G)$. 
In the $k$-fixed strategies,
as $k$ gets smaller, the \kCol\ problem must be harder. 
A simple method like the random/greedy method
must yield a poor initial solution for such $k$. 

We claim that, when $k$ is small to some extent, 
a legal $(k+1)$-coloring should be precious
in the sense that it cannot be obtained easily. 
We expect that a good $k$-coloring (in terms of the penalty function)
could be obtained by
a 
slight modification of the legal $(k+1)$-coloring. 
The expectation is supported by \tabref{extracol}.
The table shows the distribution of color class sizes
in a legal 146-coloring for the {\tt C2000.5} instance
that is found by Wu and Hao~\cite{WH.2012}.
Most of the color classes are large.
In particular, more than $1/3$ of the 146 color classes are largest independent sets
that consist of 16 vertices. 
We also see that only few color classes are small.
Even though we recolor some small color classes arbitrarily
so that there remain at most $k$ colors,
we could obtain a $k$-coloring
that has few conflicting edges.

\begin{table}[t!]
  \centering
  \caption{The distribution of color class sizes in a legal 146-coloring for the {\tt C2000.5} instance~\cite{WH.2012}; the current best known bound is 145~\cite{HW.2012}}
  \label{tab:extracol}
  \begin{tabular}{lrrrrrrrrrlr}
    \hline
    & \multicolumn{9}{c}{Size $b$} && Total \\
    & 8 & 9 & 10 & 11 & 12 & 13 & 14 & 15 & 16\\
    \hline
    Number of color\\
    classes with size $b$ & 2 & 4 & 15 & 10 & 16 & 14 & 14 & 18 & 53 && 146\\
    \hline
  \end{tabular}
\end{table}

The last motivation is theoretical.
Various configurations of the recycle method are possible,
according to how to take $K$ (and $\varepsilon$ for the $k$-fixed-penalty) and how to change colors of the vertices
in $\bigcup_{i\in K} V_{c,i}$. 
For some configurations, 
we can derive an upper bound on the penalty value
of a generated solution,
with respect to the graph size and $k$.

We denote a given legal $(k+1)$-coloring by $c$.
Without loss of generality, we assume that
$V_{c,k+1}$ is
the 
smallest color class,
that is,
$|V_{c,k+1}|\le |V_{c,i}|$ holds for any $i\in[k+1]$. 
We consider a configuration of the recycle method such that
the smallest color class $V_{c,k+1}$ is recolored,
that is, $K=\{k+1\}$. 

\paragraph{$k$-fixed-penalty.}
By $K=\{k+1\}$,
the color $\varepsilon\in K$ to be removed is $k+1$. 
Let us define $\MR(c)$ to be the set of all $k$-colorings
that can be obtained by changing the colors of vertices in $V_{c,k+1}$ to ones in $[k]$,
that is,
\begin{align*}
  \MR(c)\triangleq
  \{c'\in\Scomp{k}\mid \forall i\in[k],\ V_{c',i}\supseteq V_{c,i},\ V_{c',i}\setminus V_{c,i}\subseteq V_{c,k+1}\}.
\end{align*}
For any $c'\in\MR(c)$,
we can derive an upper bound on the penalty value $\rhocomp(c')$, as stated by the following proposition.

\begin{prop}
  \label{prop:penalty_ub}
  For a graph $G=(V,E)$ with
  the 
  maximum degree $\Delta$,
  let $c$ be a legal $(k+1)$-coloring
  such that $|V_{c,k+1}|\le |V_{c,i}|$ holds for any $i\in[k+1]$.
  Any $c'\in\MR(c)$ satisfies
  $\rhocomp(c')\le\displaystyle\frac{n\Delta}{k+1}$,
  where $n=|V|$. 
\end{prop}
\begin{proof}
  Because $\{V_{c,1},\dots,V_{c,k+1}\}$ is a partition of $V$
  and $V_{c,k+1}$ is
  the 
  smallest set among them,
  we have $|V_{c,k+1}|\le\frac{n}{k+1}$.
  For any $i\in[k+1]$,
  no edge should exist between any two vertices in $V_{c,i}$ 
  because $V_{c,i}$ is independent.
  For any $i,j\in[k]$ with $i\ne j$,
  every edge between a vertex in $V_{c,i}$ and a vertex in $V_{c,j}$
  is not conflicting in $c'$. 
  Then every conflicting edge in $c'$
  is incident to a vertex in $V_{c,k+1}$.
  Because the degree of a vertex is at most $\Delta$,
  the number of edges incident to vertices in $V_{c,k+1}$
  is at most $\frac{n\Delta}{k+1}$.
\end{proof}


\propref{penalty_ub} tells that,
even though we assign arbitrary colors in $[k]$ to $V_{c,k+1}$,
we have a $k$-coloring whose penalty value is at most $\frac{n\Delta}{k+1}$.

Suppose that we generate a $k$-coloring $c'$ by
assigning a color $i\in[k]$ to each $v\in V_{c,k+1}$,
where the color $i$ is chosen at
random.  
The expected value of $\rhocomp(c')$ is bounded as follows. 

\begin{prop}
  \label{prop:penalty_exp}
  For a graph $G=(V,E)$ with
  the 
  maximum degree $\Delta$,
  suppose that we are given a legal $(k+1)$-coloring $c$
  such that $|V_{c,k+1}|\le |V_{c,i}|$ holds for any $i\in[k+1]$.
  Let $c'\in\MR(c)$ denote a $k$-coloring that is obtained by
  assigning a color $i\in[k]$ to each $v\in V_{c,k+1}$,
  where the color $i$ is chosen from $[k]$ at
  random. 
  The expected value of $\rhocomp(c')$ is at most $\displaystyle\frac{n\Delta}{k(k+1)}$,
  where $n=|V|$. 
\end{prop}
\begin{proof}
  For a vertex $v\in V_{c,k+1}$ and a color $i\in[k]$,
  let $f_i$ denote the number of conflicting edges incident to $v$
  that appear when $v$ is assigned the color $i$.
  The expected number of conflicting edges incident to $v$
  is $f_1/k+\dots+f_k/k\le\Delta/k$.
  By $|V_{c,k+1}|\le n/(k+1)$,
  the expected number of all conflicting edges is at most $\frac{n\Delta}{k(k+1)}$.
\end{proof}

Let us compare the penalty values of
a $k$-coloring $c'$ that is obtained in the manner of \propref{penalty_exp}
and a completely random $k$-coloring $r$.
By a completely random $k$-coloring, we mean that
every vertex is assigned a color $i\in[k]$
with probability $\frac{1}{k}$.
Let $m=|E|$. 
Because each edge
is conflicting with probability $\frac{1}{k}$,
the expectation of $\rhocomp(r)$ is
$\frac{m}{k}$ from its linearity. 
On the other hand, the expectation of $\rhocomp(c')$
is {\em at most\/} $\frac{n\Delta}{k(k+1)}$. 
When $k$ is large to some extent,
it is highly likely that $\rhocomp(c')<\rhocomp(r)$ holds. 

The following proposition provides derandomization of \propref{penalty_exp}. 

\begin{prop}
  For a graph $G=(V,E)$ with
  the 
  maximum degree $\Delta$,
  suppose that we are given a legal $(k+1)$-coloring $c$
  such that $|V_{c,k+1}|\le |V_{c,i}|$ holds for any $i\in[k+1]$.
  We can generate a $k$-coloring $c'\in\MR(c)$ such that $\rhocomp(c')\le\displaystyle\frac{n\Delta}{k(k+1)}$
  in $O(n\Delta)$ time, where $n=|V|$. 
\end{prop}
\begin{proof}
  Because $c$ is a legal $(k+1)$-coloring,
  any neighbor of a vertex $v\in V_{c,k+1}$
  is assigned a color in $[k]$.
  There is a color $i$ in $[k]$ that appears at most $\frac{\Delta}{k}$
  times among the $v$'s neighborhood.
  To recolor $v$,
  we let $c'(v)\gets i$, which counts at most $\frac{\Delta}{k}$ conflicting edges. 
  Since $|V_{c,k+1}|\le\frac{n}{k+1}$, 
  we have $\rhocomp(c')\le\frac{n\Delta}{k(k+1)}$.
  The time complexity is obvious. 
\end{proof}


\paragraph{$k$-fixed-partial.} 
By $K=\{k+1\}$,
we generate a partial $k$-coloring $c'$
from a legal $(k+1)$-coloring $c$
by assigning the dummy color $\phi$ to
the 
smallest color class $V_{c,k+1}$,
whereas the other color classes remain the same. 
We can derive an upper bound on the penalty value
of the generated $k$-coloring.

\begin{prop}
  \label{prop:partial_ub}
  For a graph $G=(V,E)$ with
  the 
  maximum degree $\Delta$,
  let $c$ be a legal $(k+1)$-coloring
  such that  $|V_{c,k+1}|\le |V_{c,i}|$ holds for any $i\in[k+1]$
  and $c'$ be a partial $k$-coloring
  such that $c'(v)=c(v)$ for all $v\in V\setminus V_{c,k+1}$ and
  $c'(v)=\phi$ for all $v\in V_{c,k+1}$. 
  It holds that $\rhopart(c')\le\displaystyle\frac{n}{k+1}$,
  where $n=|V|$. 
\end{prop}
\begin{proof}
  The bound is due to $|V_{c,k+1}|\le\frac{n}{k+1}$. 
\end{proof}

Let us give a remark on the penalty value of a random partial $k$-coloring, say $r'$,
which is constructed as follows;
we first assign a random color from $[k]$ to every vertex in $V$.
Then, while there is a conflicting edge,
we repeat removing the color of a conflicting vertex
(i.e., $\phi$ is assigned to the vertex). 
Because the expected number of conflicting edges is $\frac{m}{k}$,
the expected penalty value $\rhopart(r')$ is at most $\frac{m}{k}$.
We regard that the partial $k$-coloring $c'$ of \propref{partial_ub} has a better upper bound
since $\rhopart(c')\le\frac{n}{k+1}$ {\em always} holds
and $\frac{n}{k+1}\le\frac{m}{k}$
if $n\le m$, which holds when $G$ is connected and not a tree. 

\section{Computational Studies}
\label{sec:comp}

In this section, we present computational results
to show how the recycle method is effective,
in comparison with conventional initial solution generators.
We observe how the performance of three $k$-fixed strategic heuristics changes
with respect to initial solution generators. 
The heuristics include two local search algorithms,
\Tabucol~\cite{HW.1987} and \Partcol~\cite{BZ.2008},
and a modern hybrid method, \Head~\cite{MG.2018}. 

We mention the experimental setup in \secref{comp.setup}.
In Sections~\ref{sec:comp.local} and \ref{sec:comp.head},
we present the results on the local search algorithms
and the hybrid method, respectively.

Our purpose in the experiment is not to update the best-known upper bounds of $\chi(G)$ for benchmark instances
but is to show how the recycle method is effective in comparison with conventional initial solution generators. 
The best known bounds could be updated if we run heuristic algorithms for tens/hundreds of hours, as in \cite{BZ.2008,MG.2018},
with carefully determined parameter values. 
However, it is not our interest here.
We would like to show that, with an initial solution generated by the recycle method, a $k$-fixed strategic
heuristic algorithm finds a legal $k$-coloring more quickly especially when $k$ is not sufficiently small.
This indicates that the recycle method enables us to share more time
to work on the \kCol\ problem for smaller $k$.

\subsection{Experimental settings}
\label{sec:comp.setup}

All the experiments are conducted on a workstation that carries
an Intel Core i7-4770 Processor (up to 3.90GHz by means of Turbo Boost Technology)
and 8GB main memory. The installed OS is Ubuntu 16.04.

For benchmark instances, we take up two types of instances:
DIMACS instances
and Carter et al.'s timetabling instances~\cite{CLL.1996}.
We let each heuristic algorithm
solve an instance 50 times with different random seeds. 

Suppose that we have found a legal $(k+1)$-coloring $c$ for some integer $k$.
The next step of the
$k$-fixed strategy 
is to solve the \kCol\ problem.
We would like to generate an initial solution for heuristic algorithms.  
We summarize the
five 
initial solution generators as follows.

\begin{itemize}
\item Recycle method ({\sc \Re}): 
    We use the recycle method in the configuration of \propref{penalty_exp}.
    Let $i^\ast\in[k+1]$ denote a color such that  
    $V_{c,i^\ast}$ is
    the 
    smallest color class among $V_{c,1},\dots,V_{c,k+1}$. 
    We set $K=\{i^\ast\}$. 
    In the $k$-fixed-penalty strategy,
    $\varepsilon$ is automatically set to $i^\ast$.
    Vertices in $V_{c,i^\ast}$ are recolored to
    ones in $[k+1]\setminus\{\varepsilon\}$ randomly ($k$-fixed-penalty)
    or to $\phi$ ($k$-fixed-partial). 

  \item Greedy method ({\sc \Gr}): 
    We visit the vertices in a random order.
    For a visited vertex $v$,
    if there is a color $i\in[k]$ such that
    assigning $i$ to $v$ does not produce a conflicting edge,
    then we assign the smallest $i$ to $v$.
    Otherwise, we assign a random color in $[k]$ to $v$ ($k$-fixed-penalty)
    or leave $v$ uncolored ($k$-fixed-partial). 
      
  \item Random method ({\sc \Rnd}): 
    In the $k$-fixed-penalty strategy,
    we assign a random color in $[k]$ to each vertex $v\in V$.
    In the $k$-fixed-partial strategy,
    we visit vertices in a random order.
    Then for each $v\in V$, 
    let us denote by $K_v$ the set of colors that
    appear in the neighborhood of $v$.
    If $K_v\subsetneq[k]$, then
    we assign a random color in $[k]\setminus K_v$ to $v$.
    Otherwise (i.e., no color is available),
    we assign the dummy color $\phi$ to $v$.

  \item\Dsatur~\cite{B.1979} ({\sc \Dsa}):
    It is a constructive algorithm
    that repeats picking up
    a 
    vertex with the largest
    chromatic degree
    (i.e., the number of colors that appear in the adjacent vertices),
    and then assigning the smallest color to the vertex
    so that no conflicting edge arises,
    until all vertices are assigned colors.
    Suppose that we have obtained a $k'$-coloring $(k'>k)$
    by \Dsa. To make it a $k$-coloring, 
    we turn the color of any vertex
    that is assigned a color in $\{k+1,\dots,k'\}$ into
    one in $[k]$ randomly ($k$-fixed-penalty)
    or into the dummy color $\phi$ ($k$-fixed-partial).
    
  \item\Rlf~\cite{L.1979}:
    It repeats extracting a maximal independent set
    as a color class until all vertices are chosen.
    Suppose that $k$ independent sets
    have been extracted as color classes by \Rlf.
    To obtain a $k$-coloring based on them,
    we assign random colors in $[k]$ ($k$-fixed-penalty) or
    the dummy color $\phi$ ($k$-fixed-partial)
    to the remaining vertices.     
\end{itemize}

\subsection{Recycle method on local search algorithms}
\label{sec:comp.local}

\paragraph{Overview of the algorithms.}
\Tabucol~\cite{HW.1987} and \Partcol~\cite{BZ.2008} are local search algorithms
that exploit the $k$-fixed-penalty and
$k$-fixed-partial 
strategies, respectively. 

We downloaded the source codes of \Tabucol\ and \Partcol\ (written in C)
from R.~Lewis's website.\footnote{\url{http://rhydlewis.eu/resources/gCol.zip}}
In the program, $\Gr$ is implemented as the initial solution generator for the \kCol\ problem
(line 4 in \algref{k-fixed}). 
We appended implementation of the other four generators in the source codes,
utilizing functions and data structures in the original source code as possible. 

For the iterative scheme (\algref{k-fixed}),
we construct the first initial solution in line 1
by \Dsatur~\cite{B.1979}. 
%
For the termination condition in the while-loop, 
we set the upper limit of computation time to 600 seconds.
In our preliminary studies, we observed that
computation time of the initial solution generators
is negligible in comparison with the tabu search. 
We also observed that 600 seconds are enough to
draw our conclusion that the recycle method is effective
especially when $k$ is not sufficiently small. 
Some of other papers in the literature
use a much longer time limit (e.g., ten hours in \cite{BZ.2008}).
Their experimental purpose is to update the best-known upper bound of a chromatic number,
which is different from our experiments. 

Both \Tabucol\ and \Partcol\ are tabu search algorithms.
The tabu tenure is determined by either the dynamic scheme (\Dyn)
or the reactive scheme (\Foo). 
\begin{itemize}
\item\Dyn: The tabu tenure is set to $\alpha n_c+\gamma$,
where $\alpha=0.6$, $n_c$ is the number of conflicting vertices,
and $\gamma$ is an integer
that is picked up from $\{0,\dots,9\}$ at 
random.
This setting of the tabu tenure
is recommended in the literature~\cite{BZ.2008,GH.1999,LTMG.2012}. 
\item\Foo: Proposed in \cite{BZ.2008}, the scheme provides a reactive tabu tenure based on the
  fluctuation of the objective (\Foo) function.
  Roughly, if the penalty value does not change during a long period,
  the tenure is set to a large value in order to escape from search stagnation.
  It is then decreased slowly along the search process. 
  We use random parameter values, following the R.~Lewis's original implementation. 
\end{itemize}
We represent \Tabucol\ whose tabu tenure is set by the \Dyn\ scheme
as \Dyn-\Tabucol.
It is analogous with \Foo-\Tabucol, \Dyn-\Partcol, and \Foo-\Partcol.
We will present the results of only \Dyn-\Tabucol\ and
\Foo-\Partcol\ in the paper.
The other two algorithms achieve
the similar results.


\paragraph{DIMACS instances.}
We solve 20 DIMACS instances that are treated in \cite{BZ.2008}.
We present the results of \Dyn-\Tabucol\ and \Foo-\Partcol\
in Tables \ref{tab:summary} and \ref{tab:summary_partcol}, respectively. 
For each instance,
the number after the first alphabets represents the number of vertices.
For example,
the 
{\tt DSJC1000.5} instance consists of 1000 vertices
and is a random graph of the Erd\"os-R\'enyi model with edge density 0.5. 
See \cite{BZ.2008} for
a 
description of the instances.

\begin{table}
  \centering
  \caption{Results of \Dyn-\Tabucol\ for the selected 20 DIMACS instances}
  \label{tab:summary}
  {\small
  \begin{tabular}{lr|rrrrr|rr}
    \hline
    Instance &
    \multicolumn{1}{c|}{First} &
    \multicolumn{1}{c}{\Re} &
    \multicolumn{1}{c}{\Gr} &
    \multicolumn{1}{c}{\Rnd} &
    \multicolumn{1}{c}{\Dsa} &
    \multicolumn{1}{c|}{\Rlf} &
    \multicolumn{1}{c}{Best} &
    \multicolumn{1}{c}{$\chi$} \\
    &&&&&&& \multicolumn{2}{c}{\cite{GHHP.2013,MT.2010,MG.2018}}\\
    \hline
    {\tt DSJC1000.1}& 25 (1)  & 20 (6)& 20 (4)& 20 (6)& \bf 20 (11)& 20 (5)& 20 & $\ge10$\\
    {\tt DSJC1000.5}& 113 (1)  & 88 (3)& 88 (2)& 88 (3)& 89 (37)& 88 (1)& 82 & $\ge73$\\
    {\tt DSJC1000.9}& 297 (1)  & \bf 225 (6)& 226 (32)& 226 (31)& 226 (19)& 225 (1)& 222 & $\ge216$\\
    {\tt DSJC500.1}& 15 (11)  & 12 (50)& 12 (50)& 12 (50)& 12 (50)& 12 (50)& 12 & $\ge5$\\
    {\tt DSJC500.5}& 63 (3)  & \bf 49 (13)& 49 (8)& 49 (3)& 49 (7)& 49 (9)& 47 & $\ge43$\\
    {\tt DSJC500.9}& 161 (3)  & 126 (2)& 126 (1)& \bf 126 (3)& 127 (50)& 127 (50)& 126 & $\ge123$\\
    {\tt DSJR500.1c}& 87 (3)  & 86 (11)& 87 (3)& 87 (3)& 86 (7)& \bf 85 (1)& 85 & 84\\
    {\tt DSJR500.5}& 129 (9)  & \bf 124 (4)& 127 (1)& 128 (3)& 125 (7)& 126 (2)& 122 & 122\\
    {\tt R1000.1c}& 103 (3)  & \bf 98 (33)& 98 (11)& 98 (12)& 98 (23)& 98 (24)& 98 & 98 \\
    {\tt R1000.5}& 250 (50)  & \bf 240 (4)& 247 (1)& 249 (9)& 240 (2)& 243 (2)& 234 & 234\\
    {\tt R250.1c}& 65 (34)  & 64 (26)& 65 (34)& 65 (34)& \bf 64 (28)& 64 (12)& 64 & 64\\
    {\tt R250.5}& 66 (8)  & 66 (43)& 66 (11)& 66 (8)& \bf 66 (46)& 66 (8)&65&65 \\
    {\tt flat1000\_50\_0}& 111 (1)  & \bf 50 (50)& 56 (2)& 56 (1)& 58 (1)& 58 (1)&50&50\\
    {\tt flat1000\_60\_0}& 112 (4)  & \bf 60 (50)& 73 (2)& 74 (5)& 75 (3)& 75 (1)&60&60\\
    {\tt flat1000\_76\_0}& 112 (1)  & 87 (1)& \bf 87 (2)& 87 (1)& 87 (1)& 88 (47)&81&76\\
    {\tt flat300\_28\_0}& 40 (1)  &\bf 28 (4)& 30 (1)& 30 (1)& 31 (49)& 31 (50)&28&28\\
    {\tt le450\_15c}& 23 (12)  & 15 (2)&\bf 15 (3)& 15 (1)& 15 (2)& 15 (1)&15&15\\
    {\tt le450\_15d}& 23 (2)  & 16 (50)& 16 (49)& 16 (50)& 15 (1)& 16 (50)&15&15\\
    {\tt le450\_25c}& 28 (4)  & 26 (50)& 26 (50)& 26 (50)& 26 (49)& 26 (50)&25&25\\
    {\tt le450\_25d}& 28 (19)  & 26 (50)& 26 (50)& 26 (50)& 26 (50)& 26 (50)&25&25\\
    \hline
    \multicolumn{2}{c|}{Averaged rank} & 1.40 & 2.85 & 2.90 & 2.80 & 2.95 & - & - \\
    \hline
  \end{tabular}
  }
\end{table}

\begin{table}
  \centering
  \caption{Results of \Foo-\Partcol\ for the selected 20 DIMACS instances}
  \label{tab:summary_partcol}
  {\small
    \begin{tabular}{lr|rrrrr|rr}
      \hline
      Instance &
      \multicolumn{1}{c|}{First} &
      \multicolumn{1}{c}{\Re} &
      \multicolumn{1}{c}{\Gr} &
      \multicolumn{1}{c}{\Rnd} &
      \multicolumn{1}{c}{\Dsa} &
      \multicolumn{1}{c|}{\Rlf} &
      \multicolumn{1}{c}{Best} &
      \multicolumn{1}{c}{$\chi$} \\
      &&&&&&& \multicolumn{2}{c}{\cite{GHHP.2013,MT.2010,MG.2018}}\\
      \hline
          {\tt DSJC1000.1}& 25 (1)  & \bf 20 (1)& 21 (50)& 21 (50)& 21 (50)& 21 (50)& 20 & $\ge10$\\
          {\tt DSJC1000.5}& 113 (1)  & \bf 89 (20)& 89 (13)& 89 (18)& 89 (12)& 89 (14)& 82 & $\ge73$\\
          {\tt DSJC1000.9}& 297 (1)  & \bf 227 (21)& 229 (15)& 229 (17)& 229 (3)& 230 (41)& 222 & $\ge216$\\
          {\tt DSJC500.1}& 15 (11)  & 12 (50)& 12 (50)& 12 (50)& 12 (46)& 12 (50)& 12 & $\ge5$\\
          {\tt DSJC500.5}& 63 (3)  & 49 (7)& \bf 49 (13)& 49 (11)& 49 (10)& 49 (12)& 47 & $\ge43$\\
          {\tt DSJC500.9}& 161 (3)  & \bf 127 (14)& 127 (6)& 127 (11)& 127 (8)& 127 (9)& 126 & $\ge123$\\
          {\tt DSJR500.1c}& 87 (3)  & 85 (11)& 85 (7)& 85 (12)& 85 (7)& \bf 85 (17) & 85 & 84\\
          {\tt DSJR500.5}& 129 (9)  & \bf 126 (15)& 127 (3)& 127 (3)& 126 (13)& 126 (11)&122&122\\
          {\tt R1000.1c}& 103 (3)  & \bf 98 (3)& 99 (16)& 99 (2)& 99 (16)& 99 (22)&98&98\\
          {\tt R1000.5}& 250 (50)  & \bf 243 (4)& 250 (50)& 250 (50)& 244 (23)& 246 (5)&234&234\\
          {\tt R250.1c}& 65 (34)  & 64 (50)& 64 (48)& 64 (50)& 64 (50)& 64 (50)&64&64\\
          {\tt R250.5}& 66 (8)  & \bf 66 (26)& 66 (10)& 66 (9)& 66 (24)& 66 (12)&65&65\\
          {\tt flat1000\_50\_0}& 111 (1)  & \bf 50 (50)& 78 (7)& 78 (6)& 79 (3)& 78 (1)&50&50\\
          {\tt flat1000\_60\_0}& 112 (4)  & \bf 60 (50)& 84 (1)& 85 (26)& 85 (7)& 85 (4)&60&60\\
          {\tt flat1000\_76\_0}& 112 (1)  &  88 (20)& 88 (14)& 88 (18)& 88 (13)& \bf88 (27)&81&76\\
          {\tt flat300\_28\_0}& 40 (1)  & \bf 28 (37)& 28 (3)& 29 (6)& 28 (4)& 28 (1)&28&28\\
          {\tt le450\_15c}& 23 (12)  & 15 (50)& 15 (50)& 15 (50)& 15 (50)& 15 (50)&15&15\\
          {\tt le450\_15d}& 23 (2)  & 15 (50)& 15 (50)& 15 (50)& 15 (50)& 15 (50)&15&15\\
          {\tt le450\_25c}& 28 (4)  & 26 (4)& 26 (1)& \bf 26 (5)& 26 (3)& 26 (3)&25&25\\
          {\tt le450\_25d}& 28 (19)  & 26 (3)& 26 (4)& \bf 26 (7)& 26 (6)& 26 (2)&25&25\\
          \hline
          \multicolumn{2}{c|}{Averaged rank} & 1.55 & 3.05 & 2.50 & 3.05 & 2.65 & - & - \\
          \hline
    \end{tabular}
    }
\end{table}

In the column ``First'' (resp., ``\Re '' to ``\Rlf ''),
the integer without parentheses in the 
entry indicates
the smallest number of colors attained by the first
initial solution generator
(resp., by \Dyn-\Tabucol\ or \Foo-\Partcol\ with \Re\ to \Rlf) over 50 trials,
and the integer in parentheses
indicates how many trials the smallest number of colors
is achieved. 
%
%
In the rightmost two columns,
we show
the 
best-known upper bounds (``Best'')
and
the 
chromatic numbers or their lower bounds (``$\chi$'')
that are reported in recent papers~\cite{GHHP.2013,MT.2010,MG.2018}.
They are shown just for readers' information.
Recall that our interest here is not in improving the bounds
but in the performance difference derived from initial solution generators.
%

We compare the five initial solution generators in 
the smallest number of colors achieved and its frequency
over 50 trials. 
For each instance,
we regard that an initial solution generator $g$
is more effective than another generator $g'$ if
\begin{itemize}
\item $g$ achieves a strictly smaller number of colors than $g'$; or
\item $g$ and $g'$ achieve the same number of colors, but $g$ does it more frequently. 
\end{itemize}
In Tables \ref{tab:summary} and \ref{tab:summary_partcol},
we indicate the only best generator
(i.e., no tie with or lost to other generator) for each instance
by boldface if one exists.
We see that \Re, \Gr, \Rnd, \Dsa\ and \Rlf\ 
are the only best generators in 8, 2, 1, 3 and 1 instances in \tabref{summary},
and in 10, 1, 2, 0, 2 instances in \tabref{summary_partcol},
respectively.
We see that, for both local search algorithms,
\Re\ is the best initial solution generator in most instances. 
We also show 
the averaged rank of each initial solution generator
over the instances in the last row of the tables. 
We see that \Re\ is much better than the others
while the others are more competitive. 

The results tell that, with \Re, we are more likely to obtain
a legal $k$-coloring with a smaller $k$. 
We consider that it is thanks to better initial solutions
generated by the recycle method,
by which \Dyn-\Tabucol\ or \Foo-\Partcol\ solves
the \kCol\ problem more quickly.

\paragraph{Timetabling instances.}
Next, we solve 13 timetabling instances
posed by Carter et al.~\cite{CLL.1996}.
The instances ask to assign time slots (colors) to exams (vertices)
so that no two exams are assigned the same time slot
if there is a student who takes both exams
(i.e., the two exams are joined by an edge). 
We are required to use as few time slots as possible. 
According to \cite{LTMG.2012},
the timetabling instances are different from random graphs
in degree coefficients of variation (CV).
The degree CV is defined as the ratio
of the standard deviation over the averaged vertex degree.  
In a random graph, it is shown that the degree CV does
not exceed $28\%$,
while it is from $36.3\%$ to $129.5\%$ in the timetabling instances.
This means that the vertex degree is more scattering
in the timetabling instances and that the timetabling instances
has different graph structures from DIMACS instances, most of which are
generated as random graphs. 

We show the results of \Dyn-\Tabucol\ in \tabref{timetable}.
We see
the 
superiority of the recycle method again.
We see that \Re, \Gr, \Rnd, \Dsa\ and \Rlf\
are the only best generators in 5, 0, 0, 2 and 1 instances, respectively. 
%
The averaged rank in the last row shows
that \Re\ is the best, followed by \Dsa\ and \Rlf,
which are comparatively better than \Gr\ and \Rnd. 

\begin{table}[t!]
  \centering
  \caption{Results of \Dyn-\Tabucol\ for the 13 timetabling instances} 
  \label{tab:timetable}
  {\small
  \begin{tabular}{rrr|rrrrr|rr}
    \hline
    Instance &
    \multicolumn{1}{c}{$n$} &
    \multicolumn{1}{c|}{First} &
    \multicolumn{1}{c}{\Re} &
    \multicolumn{1}{c}{\Gr} &
    \multicolumn{1}{c}{\Rnd} &
    \multicolumn{1}{c}{\Dsa} &
    \multicolumn{1}{c|}{\Rlf} &
    \multicolumn{1}{c}{\cite{LTMG.2012}} \\
    \hline
        {\tt hec-s-92}&81& 19 (50)  & \bf 17 (49)& 17 (37)& 17 (37)& 17 (48)& 17 (43) & 17\\
        {\tt sta-f-83}&139& 13 (50)  & 13 (50)& 13 (50)& 13 (50)& 13 (50)& 13 (50) & 13\\
            {\tt yor-f-83}&181& 20 (50)  & \bf 19 (38)& 19 (22)& 19 (22)& 19 (37)& 19 (27) & 19\\
            {\tt ute-s-92}&184& 10 (50)  & 10 (50)& 10 (50)& 10 (50)& 10 (50)& 10 (50) & 10\\
            {\tt ear-f-83}&190& 23 (50)  & \bf 22 (48)& 23 (50)& 23 (50)& 22 (47)& 22 (45) & 22\\
            {\tt tre-s-92}&261& 23 (50)  & 20 (28)& 20 (25)& 20 (26)& 20 (29)& \bf 20 (35) & 20\\
            {\tt lse-f-91}&381& 19 (50)  & 17 (50)& 17 (1)& 18 (18)& 17 (50)& 17 (9) & 17\\
            {\tt kfu-s-93}&461& 19 (50)  & 19 (50)& 19 (50)& 19 (50)& 19 (50)& 19 (50) & 19\\
            {\tt rye-s-93}&486& 22 (50)  & \bf 21 (37)& 21 (8)& 21 (8)& 21 (31)& 21 (11) & 21\\
            {\tt car-f-92}&543& 30 (50)  & 28 (43)& 30 (50)& 30 (50)& \bf 28 (49)& 28 (8) & 27\\
            {\tt uta-s-92}&622& 31 (50)  & 31 (50)& 31 (50)& 31 (50)& 31 (50)& 31 (50) & 29\\
            {\tt car-s-91}&682& 31 (50)  & \bf 28 (14)& 30 (4)& 30 (1)& 28 (9)& 28 (12) & 28\\
            {\tt pur-s-93}&2419& 35 (50)  & 33 (44)& 35 (50)& 35 (50)& \bf 33 (47)& 34 (18) & 33\\
            \hline
            \multicolumn{3}{c|}{Averaged rank} & 1.31 & 3.15 & 3.23 & 1.54 & 2.15 & -  \\
            \hline
  \end{tabular}
  }
\end{table}

In the rightmost column, we show the best results in \cite{LTMG.2012} for readers' information,
where \Tabucol, \Partcol, HEA~\cite{GH.1999}, ant-colony method~\cite{DT.2008},
hill-climbing method~\cite{L.2009} and backtracking \Dsatur\ algorithm~\cite{B.1979} are compared.

\paragraph{Discussion.}
We have seen that the recycle method is more effective than
other initial solution generators for many instances.

The merit of using the recycle method
is that we can start the search from a good initial solution
that has a small penalty value, where
the penalty value has a theoretical upper bound. 
The remarkable success of the recycle method in 
such DIMACS instances as 
{\tt flat1000\_50\_0}
and {\tt flat1000\_60\_0} must be due to the merit. 
%

The recycle method is not necessarily effective for all $k$.
The smaller $k$, the more difficult 
the \kCol\ problem would be. 
The search space must have a large number of valleys
of locally optimal solutions. 
Then the local search that starts from a good initial solution
(i.e., one generated by the recycle method)
may fall into such valleys.
We consider that, for such $k$,
the recycle method is not effective any more. 
In DIMACS instances {\tt DSJC500.1}, {\tt DSJC500.5} and {\tt DSJC500.9}
and timetabling instances {\tt sta-f-83} and {\tt tre-s-92}, 
the local search succeeds in improving $k$ to a sufficiently small value
(e.g., the best-known value or a value close to it)
within the time limit,
regardless of the initial solution generator. 
For such $k$, the recycle method does not outperform others significantly.

\paragraph{Detailed comparison of initial solution generators.}
Here we compare the five initial solution generators in the following aspects:
\begin{description}
\item[(I)] How small is a penalty value of the initial solution?
\item[(II)] How fast does \Dyn-\Tabucol\ decrease $k$?
\end{description}

\noindent
    {\bf (I)}
    We show that the recycle method generates a good initial solution in terms of the penalty value.
    
See \figref{k-penalty}.
We show how penalty values (average over 50 trials, vertical axis) of initial solutions change
with respect to the number $k$ of colors (horizontal axis).
The figure is taken from the experiment on
the 
{\tt DSJC1000.5} instance. 
The five curves correspond to the five initial solution generators.
    
\begin{figure}[t!]
  \centering
  \includegraphics[width=11cm,keepaspectratio]{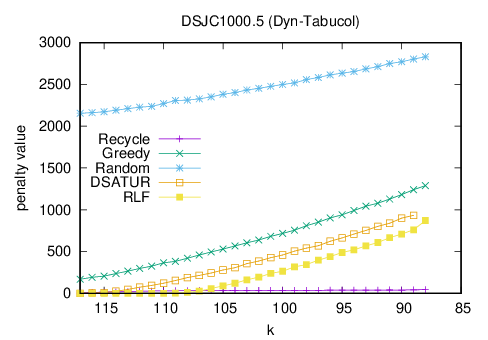}
  \caption{Penalty values of initial solutions for the {\tt DSJC1000.5} instance} 
  \label{fig:k-penalty}
\end{figure}

The figure shows the general tendency that is observed in many instances;
when $k$ gets smaller, 
the penalty values achieved by \Gr, \Rnd, \Dsa\ and \Rlf\ 
are increasing more and more,
whereas the penalty values achieved by \Re\ do not make a remarkable increase,
as depicted in \figref{k-penalty}.
%
In general, \Re\ yields the smallest penalty value,
followed by \Rlf, \Dsa, \Gr\ and \Rnd,
where \Dsa\ is better than \Rlf\ 
in some instances. 

This tendency is not observed for all instances.  
Let $k_0$ denote the number of colors in the first initial solution
that is produced by \Dsatur\ (i.e., line 1 in \algref{k-fixed}).
When $k$ is about $k_0-1$ to $k_0-5$,
\Rlf\ and \Dsa\ often yield a smaller penalty value than \Re.
This must be because the \kCol\ problem for such $k$ is so easy
that \Rlf\ and \Dsa\ can produce good initial solutions.
For instances such that the difference between $k_0$ and the achieved
number of colors is small (i.e., {\tt DSJC500.1}, {\tt R250.1c} and {\tt R250.5}),
we have not observed the tendency that is shown in \figref{k-penalty}. 

Recall the analysis on a penalty value in \secref{recycle}.
We see that, in \figref{k-penalty}, the penalty values achieved by \Rnd\ are close to the expectation $m/k$;
The instance has approximately $m=2.5\times10^5$ edges,
and when $k=10^2$, for example, the penalty value is close to $m/k=2.5\times10^3$.
On the other hand,
the expectation of a penalty value achieved by \Re\ is bounded by $\frac{n\Delta}{k(k+1)}$ (\propref{penalty_exp});
The instance has $n=10^3$ vertices and thus
the maximum degree $\Delta$ is at most $n=10^3$. 
For any $k$ that is shown in the figure, 
we see that the averaged penalty value of \Re\
is surely lower than the upper bound $\frac{n\Delta}{k(k+1)}$.

\noindent
    {\bf (II)}
Here we show that, with the recycle method, the tabu search algorithms decrease
the number $k$ of colors faster than other generators in many instances. 

In \figref{time-k}, we show how the number $k$ of colors 
(vertical axis)
is improved by \Dyn-\Tabucol\ along with computation time (average over 50 trials, horizontal axis).
We mentioned that, when \Dyn-\Tabucol\ is employed, the recycle method is the best generator
in 8 out of 20 DIMACS instances (\tabref{summary}). 
In these instances, we observe that \Dyn-\Tabucol\ with \Re\ 
decreases $k$ faster than other generators. 
We show in \figref{time-k} {\tt DSJC1000.9}, {\tt R1000.5} and {\tt flat1000\_50\_0} as examples. 
We see that improvement of \Re\  is made faster than other generators. 
On the other hand, we do not observe this tendency for instances such that the recycle method is less effective, as shown in the case
of the {\tt DSJC1000.5} instance in \figref{time-k}.

\begin{figure}[p!]
  \centering
  \begin{tabular}{cc}
    \includegraphics[width=8cm,keepaspectratio]{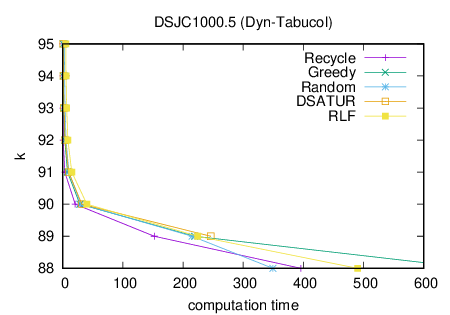}&
    \includegraphics[width=8cm,keepaspectratio]{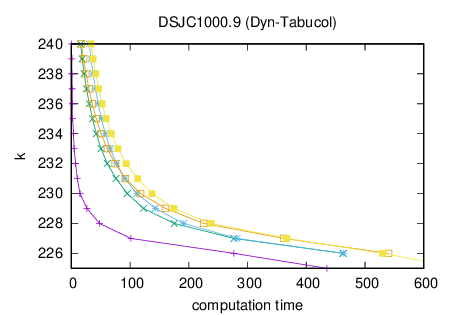}\\
    \includegraphics[width=8cm,keepaspectratio]{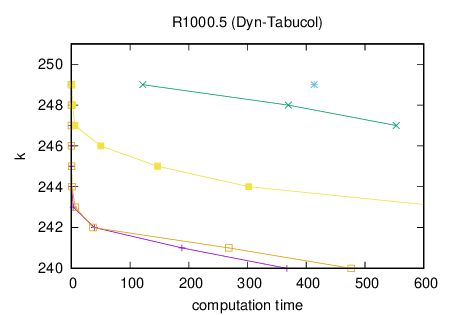}&
    \includegraphics[width=8cm,keepaspectratio]{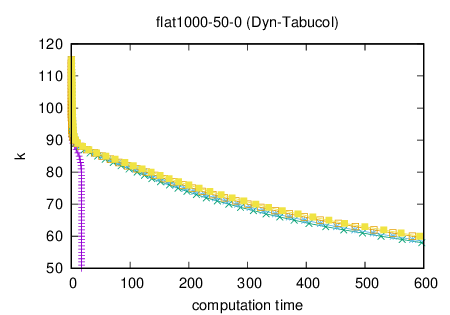}
  \end{tabular}
  \caption{How \Dyn-\Tabucol\ improves the number $k$ of colors along with computation time (s) for the {\tt DSJC1000.5}, {\tt DSJC1000.9}, {\tt R1000.5} and {\tt flat1000\_50\_0} instances}
  \label{fig:time-k}
\end{figure}

The smaller $k$,
the harder the \kCol\ problem is likely to be.
When $k$ is ``sufficiently small,''
where the extent of being ``sufficiently small'' is different
from instance to instance, 
it is not necessarily effective to start local search
from an initial solution that has a smaller penalty value
since such a solution must be among
deep valleys of locally optimal solutions.
To overcome this issue,
we should improve the heuristic algorithm itself
rather than the initial solution generator.
On the other hand, when $k$ is not so small,
there is a merit in using the recycle method as the initial solution generator
since it would decrease $k$ faster than other generators.

\subsection{Recycle method on a hybrid method}
\label{sec:comp.head}

\paragraph{Overview of the algorithm.}
The hybrid method that we take up here is
\Head\ (Hybrid Evolutionary Algorithm in Duet)~\cite{MG.2018},
which is based on
\textsc{Hea} (Hybrid Evolutionary Algorithm)~\cite{GH.1999}.
\Head\ is regarded as a {\em memetic algorithm}~\cite{M.1989},
that is, a genetic algorithm
such that each individual is improved by local search. 
The characteristic of \Head\ is that 
it maintains only two solutions in the population.

We summarize \Head\ in \algref{head}.
It is a heuristic algorithm for the \kCol\ problem
that exploits the $k$-fixed-penalty strategy,
and thus can be used within the iterative scheme in \algref{k-fixed}. 
We explain roles of important steps as follows.

\begin{algorithm}[t!]
  \caption{\Head~\cite{MG.2018}}
  \label{alg:head}
  \DontPrintSemicolon
  \SetKwInOut{Input}{Input}\SetKwInOut{Output}{Output}
  \Input{a graph $G=(V,E)$, a natural number $k$, parameters $\tau_T$ and $\tau_C$}
  \Output{a $k$-coloring for $G$}
  \SetKwRepeat{Do}{do}{while}
  
  $p_1,p_2,e_1,e_2,b\gets$ random $k$-colorings;

  $t\gets0$;

  \Do{{\rm $\rhocomp(b)>0$ and $p_1\ne p_2$}}{

    $c_1\gets$GPX$(p_1,p_2)$;

    $c_2\gets$GPX$(p_2,p_1)$;

    $p_1\gets$\Tabucol$(c_1,\tau_T)$;

    $p_2\gets$\Tabucol$(c_2,\tau_T)$;

    $e_1\gets \arg\min_{c\in\{p_1,p_2,e_1\}}\rhocomp(c)$;

    $b\gets \arg\min_{c\in\{e_1,b\}}\rhocomp(c)$;

    \If{{\rm $t$ is a multiple of $\tau_C$}}{

      $p_1\gets e_2$;

      $e_2\gets e_1$;

      $e_1\gets$ a random $k$-coloring
    }
    $t\gets t+1$
  }
  
  Output $b$
\end{algorithm}

\begin{description}
\item[Line 1:] The algorithm maintains five solutions during the execution. 
  The $p_1$ and $p_2$ are solutions in the population,
  while $e_1$ and $e_2$ are elite solutions in recent ``cycles,''
  where a cycle represents successive $\tau_C$ generations. 
  The parameter $\tau_C$ determines the length of a cycle, and we set it to the default value (10).
  The best solution in the current cycle is stored as $e_1$,
  whereas $e_2$ stores the best solution in the last cycle.
  The $b$ is used to store the best solution among those searched. 
\item[Lines 4 and 5:] GPX (Greedy Partition Crossover)
  is a crossover operator~\cite{GH.1999}.
  Given two parent solutions $p_1$ and $p_2$,
  GPX$(p_1,p_2)$ constructs a child solution $c$ by
  choosing
  the 
  largest color class from $p_1$ and $p_2$ alternatively
  as a color class of $c$.
  To be more precise,
  the 
  largest subset among $V_{p_1,1},\dots,V_{p_1,k}$ is chosen as $V_{c,1}$,
  the 
  largest subset among $V_{p_2,1}\setminus V_{c,1},\dots,V_{p_2,k}\setminus V_{c,1}$
  is chosen as $V_{c,2}$, and so forth.
  After $k$ subsets are chosen, the 
  remaining vertices are colored by 1 to $k$ randomly.
  Note that GPX$(c_1,c_2)$ and GPX$(c_2,c_1)$ can deliver different solutions.
\item[Lines 6 and 7:] \Tabucol$(c,\tau_T)$ denotes an execution of \Tabucol\ such
  that $c$ is used as the initial solution,
  $\tau_T$ is the number of iterations (where we use the default value ($3\times10^4$)),
  and the tabu tenure is determined by an original mechanism.
  The obtained solutions are used as parents in the next generation.
\item[Lines 10 to 13:]
  The variable $t$ plays the role of the generation counter.
  When $t$ is a multiple of $\tau_C$, what we call a cycle finishes. 
  In this case, a parent $p_1$ in the next generation is replaced by $e_2$,
  the best solution in the last cycle (line 11),
  $e_2$ is replaced by $e_1$, the best solution in the current cycle (line 12),
  and $e_1$ is replaced by a random $k$-coloring (line 13).
  We may say that $e_1$ and $e_2$ are used to diversify the search.
\end{description}

\paragraph{Comparison of initial population generators.}

Here we show that it is promising to use the recycle method 
in generation of the initial population.
More specifically, when the recycle method is used,
\Head\ (\algref{head}) tends to solve the \kCol\ problem more quickly,
compared with the case when the default generator is used. 

In the original implementation of \Head,
two $k$-colorings $p_1$ and $p_2$ in the initial population
are generated by \Rnd.
We call this default generation scheme by \Rnd+\Rnd. 
We compare it with \Re+\Rnd, that is, $p_1$ is generated by the recycle method
whereas $p_2$ is still generated by \Rnd.

We downloaded the source code of \Head\footnote{\url{https://github.com/graphcoloring/HEAD}} that is written in C++.
We added implementation of \Re+\Rnd\ to the source code.
We solved the \kCol\ problem 50 times for each of the 20 DIMACS instances. 
We set the time limit to 300 seconds,
that is, we break the do-while-loop
(i.e., lines 3 to 15 in \algref{head})
when the computation time exceeds 300 seconds. 

\Head\ is an algorithm for the \kCol\ problem.
We observe how fast and how often it solves the problem
within the time limit, whereas we concentrated on 
how fast local search algorithms improve $k$ in the last experiment. 

We summarize results in \tabref{pop}.
A column ``Solved'' shows how many trials
\Head\ finds a legal $k$-coloring over the 50 trials,
and a column ``Avg+Stdev'' shows the averaged computation time
and the standard deviation, which are taken
over the trials where a legal $k$-coloring is found. 
The table shows the results for $k$ such that
at least one legal $k$-coloring is found over all trials.
This means that, for example,
a legal 20-coloring is not found for {\tt DSJC1000.1}
for any trial in the experiment.

\begin{table}[t!]
  \centering
  \caption{Results of \Head\ for the selected 20 DIMACS instances}
  \label{tab:pop}
  {\small
    \begin{tabular}{lrcrrcrrc}
      \hline
      Instance &
      \multicolumn{1}{c}{$k$} &&
      \multicolumn{2}{c}{\Re$+$\Rnd} &&
      \multicolumn{2}{c}{\Rnd$+$\Rnd} &\\
      &&& Solved & Avg$\pm$Stdev && Solved & Avg$\pm$Stdev& \\
      \hline
          {\tt DSJC1000.1}
          & 22 && 50 & \bf0.033$\pm$0.006 && 50 & 0.036$\pm$0.006\\ 
          & 21 && 50 & 0.795$\pm$0.454 && 50 & 0.685$\pm$0.294\\ 
          \hline
          {\tt DSJC1000.5}
          & 85 && 50 & \bf39.10$\pm$19.45 && 50 & 46.32$\pm$14.30\\ 
          & 84 && 50 & \bf89.98$\pm$33.64 && 50 & 105.48$\pm$42.71\\ 
          & 83 && 22 & 216.68$\pm$67.42 && 16 & 207.41$\pm$47.31\\ 
          \hline
          {\tt DSJC1000.9}
          & 225 && 50 & \bf32.14$\pm$20.36 && 50 & 43.85$\pm$26.88\\ 
          & 224 && 46 & 85.81$\pm$64.98 && 45 & 99.49$\pm$51.88\\ 
          & 223 && 3 & 163.91$\pm$54.34 && 7 & 123.96$\pm$52.51\\ 
          \hline
          {\tt DSJC500.1}
          & 13 && 50 & 0.022$\pm$0.003 && 50 & 0.023$\pm$0.004\\ 
          & 12 && 50 & 36.33$\pm$26.41 && 50 & 49.12$\pm$39.17\\ 
          \hline
          {\tt DSJC500.5}
          & 50 && 50 & \bf1.05$\pm$0.519 && 50 & 1.40$\pm$0.715\\ 
          & 49 && 50 & \bf4.24$\pm$2.15 && 50 & 5.33$\pm$2.58\\ 
          \hline
          {\tt DSJC500.9}
          & 127 && 50 & 6.99$\pm$4.92 && 50 & 8.24$\pm$5.07\\ 
          & 126 && 22 & 77.46$\pm$69.62 && 19 & 84.67$\pm$70.08\\ 
          \hline
          {\tt DSJR500.1c}
          & 87 && 8 & 0.139$\pm$0.166 && 2 & 0.242$\pm$0.05\\ 
          \hline
          {\tt DSJR500.5}
          & 126 && 50 & 37.13$\pm$36.18 && 49 & 47.37$\pm$43.81\\ 
          & 125 && 40 & 118.88$\pm$86.99 && 42 & 144.36$\pm$79.31\\ 
          & 124 && 18 & 208.95$\pm$77.88 && 6 & 189.30$\pm$88.95\\ 
          \hline
          {\tt R1000.1c}
          & 99 && 15 & 3.01$\pm$7.95 && 2 & 2.37$\pm$2.23\\ 
          & 98 && \bf12 & 0.813$\pm$1.29 && 0 & N/A\\ 
           \hline
          {\tt R1000.5}
          & 249 && 37 & 135.60$\pm$59.95 && 22 & 153.25$\pm$61.02\\ 
          & 248 && 19 & 195.21$\pm$66.52 && 8 & 209.24$\pm$62.59\\ 
          & 247 && 5 & 162.87$\pm$27.94 && 1 & 268.58$\pm$0.00\\ 
          \hline
          {\tt R250.1c}
          & 65 && 7 & \bf0.047$\pm$0.025 && 3 & 0.092$\pm$0.022\\ 
          & 64 && \bf5 & 0.043$\pm$0.048 && 0 & N/A\\ 
          \hline
          {\tt R250.5}
          & 68 && 50 & \bf3.12$\pm$11.78 && 49 & 18.42$\pm$46.55\\ 
          & 67 && 43 & 26.85$\pm$45.28 && 42 & 47.83$\pm$67.37\\ 
          & 66 && 21 & 151.62$\pm$120.32 && 16 & 142.05$\pm$112.08\\ 
          \hline
          \multicolumn{9}{r}{(To be continued)}
    \end{tabular}
  }
\end{table}

\begin{table}[t!]
  \centering
  {\small
    \begin{tabular}{lrcrrcrrc}
      \multicolumn{9}{l}{(Continuation of \tabref{pop})}\\
      \hline
      Instance &
      \multicolumn{1}{c}{$k$} &&
      \multicolumn{2}{c}{\Re$+$\Rnd} &&
      \multicolumn{2}{c}{\Rnd$+$\Rnd} &\\
      &&& Solved & Avg (s) && Solved & Avg (s)& \\
      \hline
          {\tt flat1000\_50\_0}
          & 52 && 50 & \bf0.099$\pm$0.029 && 50 & 15.63$\pm$4.50\\ 
          & 51 && 50 & \bf0.113$\pm$0.038 && 49 & 19.02$\pm$6.33\\ 
          & 50 && 50 & \bf0.145$\pm$0.049 && 46 & 21.48$\pm$6.41\\ 
          \hline
          {\tt flat1000\_60\_0}
          & 62 && 50 & \bf0.243$\pm$0.075 && 50 & 28.80$\pm$7.05\\ 
          & 61 && 50 & \bf0.241$\pm$0.069 && 48 & 32.97$\pm$10.64\\ 
          & 60 && 50 & \bf0.324$\pm$0.120 && 49 & 35.40$\pm$10.22\\ 
          \hline
          {\tt flat1000\_76\_0}
          & 83 && 50 & 109.90$\pm$53.46 && 49 & 94.46$\pm$24.87\\ 
          & 82 && 20 & 209.85$\pm$37.09 && 15 & 236.64$\pm$60.47\\ 
          \hline
              {\tt flat300\_28\_0}
              & 30 && 3 & 164.94$\pm$118.42 && 1 & 315.38$\pm$0.0\\ 
              & 29 && \bf16 & 0.061$\pm$0.073 && 0 & N/A\\ 
              \hline
          {\tt le450\_15c}
          & 16 && 42 & \bf0.064$\pm$0.066 && 47 & 0.116$\pm$0.093\\ 
          & 15 && \bf1 & 0.061$\pm$0.0 && 0 & N/A\\ 
          \hline
          {\tt le450\_15d}
          & 17 && 50 & \bf0.021$\pm$0.005 && 50 & 0.060$\pm$0.014\\ 
          & 16 && 35 & 0.104$\pm$0.123 && 43 & 0.104$\pm$0.079\\ 
          \hline
          {\tt le450\_25c}
          & 27 && 50 & 0.025$\pm$0.004 && 50 & 0.026$\pm$0.004\\ 
          & 26 && 50 & 0.345$\pm$0.250 && 50 & 0.376$\pm$0.254\\ 
          \hline
          {\tt le450\_25d}
          & 27 && 50 & 0.024$\pm$0.003 && 50 & 0.025$\pm$0.003\\ 
          & 26 && 50 & 0.328$\pm$0.235 && 50 & 0.302$\pm$0.252\\ 
          \hline
    \end{tabular}
  }
\end{table}

We claim that (\Head\ with) \Re+\Rnd\ should be at least as good as \Rnd+\Rnd,
or even better than \Rnd+\Rnd\ in some cases.
Comparing the columns ``Solved'', we see that
\Re+\Rnd\ solves the \kCol\ problem
in more trials (resp., in the same number of trials
and in less trials) than \Rnd+\Rnd\ in
25 (resp., 17 and 4) out of the 46 rows in the table. 
Observe bold numbers.
They represent cases such that
\Re+\Rnd\ should be especially better than \Rnd+\Rnd\ in the following sense;
\begin{itemize}
\item The number of trials that  \Re+\Rnd\ solves the problem
  is emphasized in four rows, that is, 
12 for {\tt R1000.1c} $(k=98)$,
5 for {\tt R250.1c} $(k=64)$,
16 for {\tt flat300\_28\_0} $(k=29)$ and
1 for {\tt le450\_15c} $(k=15)$. 
In these four cases, \Rnd+\Rnd\ does not solve the problem in any trial.
\item The computation time of \Re+\Rnd\ 
is emphasized in 16 out of the 46 rows. 
In these cases, it is judged by a statistical test
that the averaged computation time is different from
(i.e., smaller than) that of \Rnd+\Rnd.
For a statistical test,
we conduct Student's $t$-test if the variance is statistically
different between two samples (decided by $F$-test)
and Welch's $t$-test otherwise, 
where the significance level is set to 0.05 in all tests.
\end{itemize}
On the other hand,
no number in the columns for \Rnd+\Rnd\ is emphasized by boldface
since \Rnd+\Rnd\ is not better than \Re+\Rnd\ in the above sense.

There are many strategies possible to generate an initial population.
We may use other constructive algorithms such as
\Gr, \Dsa\ and \Rlf\ that are mentioned in \secref{comp.setup}. 
Some solutions in the population can be generated so that
they are ``distant'' from other solutions for diversification of the search.

Still, 
we regard \Re+\Rnd\ as a reasonable
initial population generator for \Head. 
We have observed that \Re\ produces
a good initial solution in the sense of the penalty value,
compared with other constructive algorithms.
The initial population generator \Re+\Rnd\ constructs $p_1$ by \Re\ 
and $p_2$ by \Rnd, where $p_2$
is expected to be ``distant'' from $p_1$ to some extent. 
Our experiments illustrates that, for hybrid methods,
there should be merit in
including a solution that is generated by the recycle method
into an initial population.

\section{Concluding Remark}
\label{sec:conc}

In the present paper, we highlighted 
the recycle method,
an initial solution generator of a general heuristic algorithm
that employs the $k$-fixed search strategies.
We provided our motivation that includes analyses of upper bounds on the penalty value. 
Experimental results showed that
the recycle method can make 
\Tabucol\ and \Partcol, representative tabu search algorithms, 
achieve
a fewer number of colors 
than the conventional methods, that is, \Gr, \Rnd, \Dsa, and \Rlf. 
Moreover, the recycle method
accelerates improvement of
the number of colors 
in early iterations of the iterative scheme, in comparison with
these methods.
We also demonstrated that the hybrid method \Head\ with \Re+\Rnd\  
tends to find a legal $k$-coloring more quickly than \Head\ with the default generator
\Rnd+\Rnd.

The recycle method is universal in the sense that it can be applied for
any $k$-fixed strategic heuristic algorithm somehow.
Our experiments deal with three algorithms in the literature,
\Tabucol, \Partcol\ and \Head, which are just illustrative examples. 
The recycle method is so simple that it has already been used in
previous studies (e.g., \cite{CDS.2018,HK.2019}). 
However, the present paper is the first one
that investigates its effectiveness from theoretical as well as experimental viewpoints.  

The recycle method is efficient.
Let $c$ be a legal $(k+1)$-coloring.
For the configuration of $\Re$,
it takes $O(n)$ time to determine the smallest color class
and to recolor the vertices in the class,
whereas the running times of \Dsatur\ and \Rlf\
are bounded by $O(n^2)$ and $O(n^3)$, respectively~\cite{L.2016}. 
It is also worthwhile to mention that
the recycle method is easy to implement.

We hope that the recycle method becomes
a standard initial solution generator
for the \kCol\ problem.

\section*{Acknowledgments}
We gratefully acknowledge very careful and
detailed comments given by anonymous reviewers.

%

\bibliographystyle{jorsj} 
\bibliography{reference}

\end{document}